\documentclass[journal,twocolumn]{IEEEtran}
\usepackage[T1]{fontenc}
\usepackage[latin9]{inputenc}
\usepackage{amsthm}
\usepackage{amstext}
\usepackage{amssymb}
\usepackage{graphicx}
\usepackage[unicode=true,pdfusetitle,
 bookmarks=true,bookmarksnumbered=true,bookmarksopen=false,
 breaklinks=false,pdfborder={0 0 1},backref=false,colorlinks=false]
 {hyperref}
\usepackage{breakurl}
\usepackage{colortbl}

\makeatletter

\providecommand{\tabularnewline}{\\}

\theoremstyle{remark}
\newtheorem{thm}{Theorem}
\ifx\proof\undefined
\newenvironment{proof}[1][\protect\proofname]{\par
\normalfont\topsep6\p@\@plus6\p@\relax
\trivlist
\itemindent\parindent
\item[\hskip\labelsep
\scshape
#1]\ignorespaces
}{%
\endtrivlist\@endpefalse
}
\providecommand{\proofname}{Proof}
\fi
\theoremstyle{remark}
\newtheorem{lem}{Lemma}
\theoremstyle{remark}

\usepackage{amssymb}
\usepackage[cmex10]{amsmath}
\usepackage{amsthm}
\usepackage{latexsym,bm}
\usepackage{color}
\usepackage{algorithmic}
\usepackage{algorithm}

\usepackage{wasysym}

\makeatother

\begin{document}

\title{On the Blind Interference Alignment over Homogeneous Block Fading Channels}

\author{Qing F.~Zhou and Q. T. Zhang, {\em Fellow, IEEE}
%
%
\thanks{Qing F.~Zhou and Q. T. Zhang are with the Department of Electronic Engineering,
City University of Hong Kong, Tat Chee Avenue,
Kowloon, Hong Kong (Email: qinzhou@cityu.edu.hk, eekzhang@cityu.edu.hk) }
}
\maketitle

\begin{abstract}
Staggered fading pattern between different users is crucial to  interference alignment without CSIT, or so-called blind interference alignment (BIA). This special fading structure naturally arises from heterogeneous block fading setting, in which different users experience independent block fading with different coherent times. Jafar \emph{et al.} prove that BIA can be applied in some special heterogeneous block fading channels, which are formed naturally or constructed artificially. In this paper, we show that in the context of a $2$-user $2\times 1$ broadcasting (BC) channel, staggered fading pattern can also be found in homogeneous block fading setting, in which both users experience independent fading with the same coherent time; and we propose a scheme to achieve the optimal $\frac{4}{3}$ DoF for the homogenous setting by using BIA. Applying the proposed scheme, we further study a $2\times 1$ BC network with $K$ users undergoing homogeneous block fading. When $K\geq 4$, we show it is almost guaranteed that the transmitter can find two users among the $K$ users to form a $2$-user $2\times 1$ BC channel which can apply BIA.
\end{abstract}

\section{Introduction}
As a new paradigm of mitigating interference, Interference Alignment (IA) has attracted increasing attentions \cite{Cadambe2008,Choi2009,Jung2011}. In IA, the transmitters utilize the channel state knowledge to precode their signaling vectors. The aim of this precoding is that at each receiver,  interfering signals are purposely aligned in a minimized signal subspace, while the desired signals independent of interference signals are delivered in a maximized signal subspace. To implement the IA technique in practical systems, the major challenge is the need of the perfect, and sometimes global, channel state information at the transmitters (CSIT).

It is recently shown \cite{Jafar2009,Jafar2012} that the restricted condition of perfect CSIT can be relaxed if the coherent blocks of the receivers display certain heterogeneous patterns. Consider the example of a $k$-user $M\times 1$ multiple-input single-output (MISO) broadcasting (BC) channel, in which one base station with $M=2$ transmit antennas broadcasts information data to $k=2$ single-antenna users. If the coherent time $N_1$ experienced by one user is a multiple of the coherent time $N_2$ experienced by the other user, forming a specially staggered pattern, the outer bound value of $\frac{4}{3}$ degrees of freedom (DoF) is attainable by interference alignment with no CSIT. The interference alignment without CSIT is also referred to as \emph{blind interference alignment} (BIA). In order to obtain the staggered channel pattern which is amiable to BIA, Jafar \emph{et al.} \cite{Gou2011} propose to mount a configurable multi-mode antenna at each user, and artificially manipulate the channel state to construct a similar heterogeneous block fading channel by switching the antenna mode at some/all users.

In this paper, instead of applying complex receivers or presuming the heterogeneous block fading channel model, we study a $2\times 1$ BC channel with $K\geq 2$ users which follow the same temporal correlated fading model and thus have the same coherent time. To distinguish the heterogeneous block fading channel where $N_1\neq N_2$, we refer to this channel as  \emph{homogeneous} block fading model. For the homogeneous BC channel with $K=2$, we propose a scheme to achieve the optimal $\frac{4}{3}$ DoF by using BIA. Our analysis further shows that in this homogeneous BC channel, if $K \geq 4$, the chance is larger than $95\%$ that the base station find two mobile users to form a 2-user $2\times 1$ MISO BC channel, and the formed channel can achieve the optimal $\frac{4}{3}$ DoF by using BIA. When $K\geq 6$, the chance is almost one.

\section{System model}

Consider a $2$-user $2\times 1$ MISO BC channel which consists of a two-antenna transmitter and two single-antenna receivers. Referring to Fig.~\ref{fig:system.model}, let $h_{ij}$, $i,j\in\{1,2\}$, denote the channel coefficient linking the transmit antenna $\mathbb{S}_{i}$
to the receiver $\mathbb{R}_{j}$. Further define the coefficient vector $H_{j}(n)=[
h_{1j}(n), h_{2j}(n)]^{T}$ for $\mathbb{R}_{j}$, where $n\geq 0$ is the discrete-time index. Represent the coherent time for $H_{j}(n)$ by $N_{j}$, and the initial
time offset by $n_{\delta,j}$.  In this paper, we consider the BC channel with $N_1 = N_2 = N$, referred to as homogeneous block fading. Without loss of generality, let $n_{\delta,1}=0$, and $0\leq n_{\delta,2}<N$.
So, as illustrated in Fig.~\ref{fig:system.model}, the constant coefficient during a coherent block gives $H_{1}(aN+b_{1})=H'_1(a)$ for all $a \geq 0$ and $0\leq b_{1}<N$; $H_{2}(b_{2})=H'_2(0)$ for $0\leq b_{2}\leq n_{\delta,2}-1$, and $H_{2}(aN+n_{\delta,2}+b_{2})=H'_2(a+1)$
for all $a \geq 0$ and $0\leq b_{2}<N$.

It is known that the optimal achievable DoF, also known as multiplexing gain, is $\frac{4}{3}$ for $2$-user $2\times 1$ MISO BC channels \cite{Gou2011}. The interference alignment approach achieving the optimal DoF by using symbol extension is demonstrated in Fig. \ref{fig:IA.X.chnl}.
In this figure, given three arbitrary time slots $n_{1}< n_{2}< n_{3}$ --- they are not necessary to be consecutive ---
as well as the corresponding channel coefficients
\begin{center}
\begin{tabular}{|c|c|c|}
\hline
$H_{1}(n_{1})$ & $H_{1}(n_{2})$ & $H_{1}(n_{3})$\tabularnewline
\hline
$H_{2}(n_{1})$ & $H_{2}(n_{2})$ & $H_{2}(n_{3})$\tabularnewline
\hline
\end{tabular}\quad,
\par\end{center}
we denote channel coefficient matrix from node $\mathbb{S}_{i}$ to
$\mathbb{R}_{j}$ as $H_{ij}=\mathsf{diag}[h_{ij}(n_{1}),h_{ij}(n_{2}),h_{ij}(n_{3})]$.
Suppose $\mathbf{v}_{1}$ and $\mathbf{v}_{2}$ are the signaling
vectors for $\mathbb{S}_{1}$, and $\mathbf{u}_{1}$ and $\mathbf{u}_{2}$
are signaling vectors for $\mathbb{S}_{2}$. At the user $j$, the received signal $\mathbf{y}_j=[y_j(n_1),y_j(n_2),y_j(n_3)]^T$ is
\begin{equation}
\mathbf{y}_j = H_{1j} [ \mathbf{v}_1,\mathbf{v}_2 ] \mathbf{s}_1 + H_{2j}[ \mathbf{u}_1,\mathbf{u}_2 ] \mathbf{s}_2 + \mathbf{z}_j
\end{equation}
where $\mathbf{s}_i = [s_{i1},s_{i2}]^T \in \mathcal{C}^{2\times 1}$ represents two symbol streams from the transmit antenna $i$, $\mathbf{z}_j \in \mathcal{C}^{3\times 1}$ is the AWGN vector at the user $j$.
To achieve the optimal DoF $\tfrac{4}{3}$,
the interference alignment implementation shown in Fig.~\ref{fig:IA.X.chnl} requires
\begin{equation}
\begin{cases}
H_{11}\mathbf{v}_{1}\rightarrow H_{21}\mathbf{u}_{1}\\
H_{12}\mathbf{v}_{2}\rightarrow H_{22}\mathbf{u}_{2}
\end{cases}\label{eq:IA.cond}
\end{equation}
where $\mathbf{x}\rightarrow\mathbf{y}$ means that $\mathbf{x}$ aligns
with $\mathbf{y}$, that is, $\mathbf{x}=a\mathbf{y}$ for a nonzero
scalar $a$.
%
In this implementation, $\mathbb{R}_1$ decodes the symbols delivered by $\mathbf{v}_2$ and $\mathbf{u}_2$, i.e., $s_{12}$ from $\mathbb{S}_1$ and $s_{22}$ from $\mathbb{S}_2$, while $\mathbb{R}_2$ decodes the symbols $s_{11}$ and $s_{21}$. Four symbols are delivered by three channel uses, so the DoF $\frac{4}{3}$ is achieved.
Alignment criteria \eqref{eq:IA.cond} can be rewritten as
\begin{equation}
\begin{cases}
\mathbf{u}_{1} & \rightarrow \mathsf{diag}[\frac{h_{11}(n_{1})}{h_{21}(n_{1})},\frac{h_{11}(n_{2})}{h_{21}(n_{2})},\frac{h_{11}(n_{3})}{h_{21}(n_{3})}]\mathbf{v}_{1}\\
\mathbf{u}_{2} & \rightarrow \mathsf{diag}[\frac{h_{12}(n_{1})}{h_{22}(n_{1})},\frac{h_{12}(n_{2})}{h_{22}(n_{2})},\frac{h_{12}(n_{3})}{h_{22}(n_{3})}]\mathbf{v}_{2}
\end{cases}.\label{eq:IA.signal.relation}
\end{equation}

\begin{center}
\begin{figure}
\begin{centering}
\includegraphics[scale=0.48]{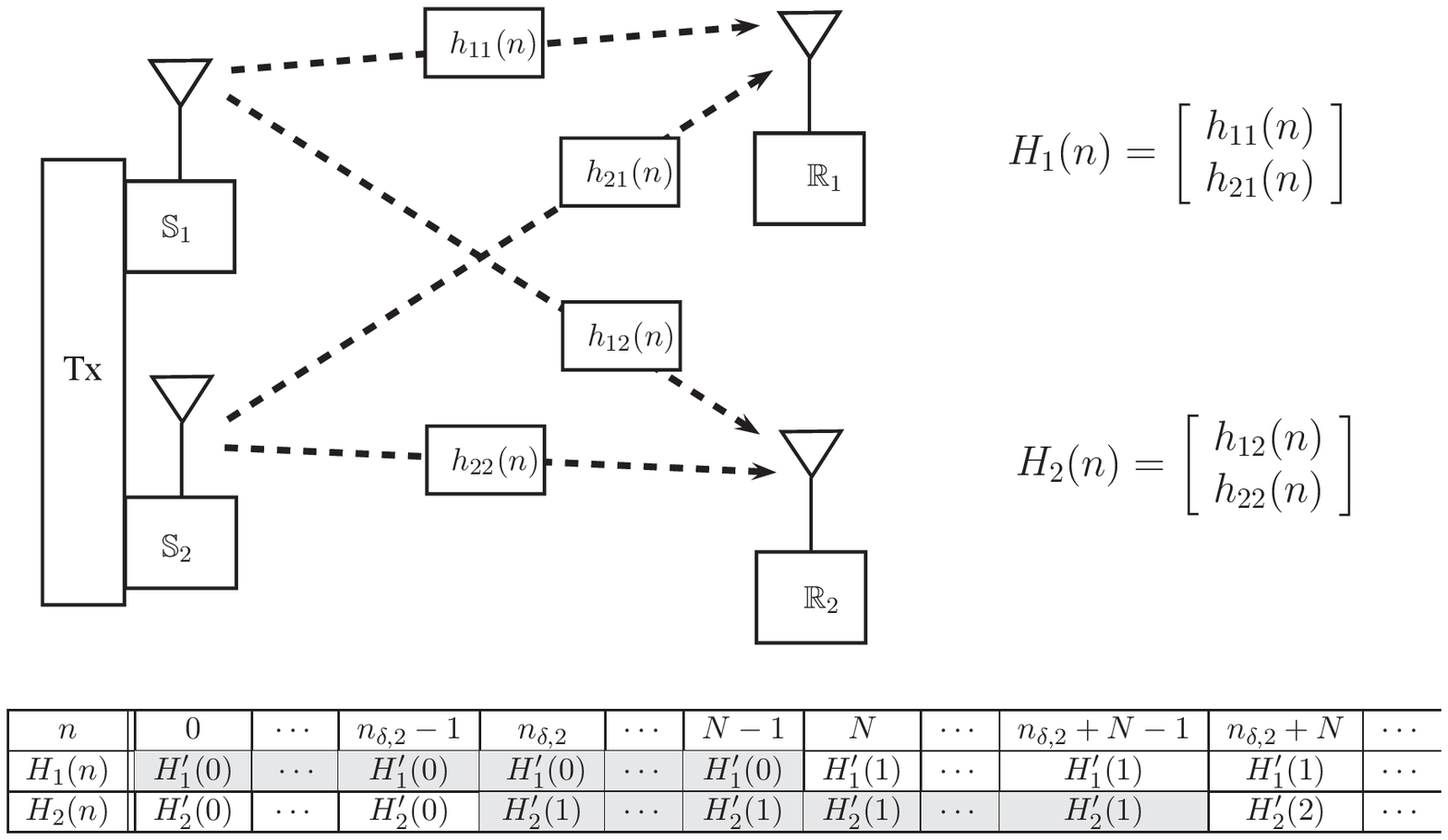}
\par\end{centering}
\caption{System model of a $2$-user $2\times 1$ MISO BC channel}
\label{fig:system.model}
\end{figure}
\par\end{center}

\begin{figure}
\begin{centering}
\includegraphics[scale=0.55]{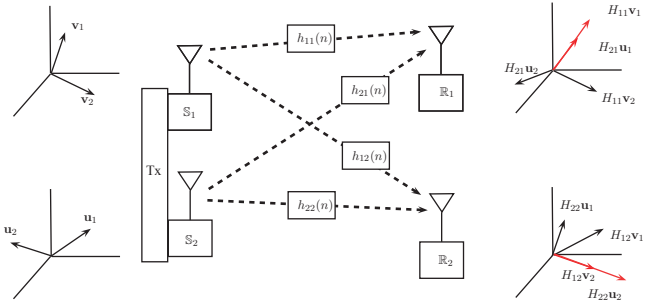}
\par\end{centering}
\caption{Interference alignment of a 2-user $2\times 1$ MISO BC channel. Note that the symbol streams $s_{ij}$ are omitted for simplicity.}
\label{fig:IA.X.chnl}
\end{figure}

\section{Block Patterns for blind interference alignment}

In this section, we investigate how to achieve interference alignment
with no CSIT but the information of channel
coherence time $N$ and the offset $n_{\delta,2}$ at the transmitter. We emphasize here that the limited knowledge of coherent time and relative time offset at the transmitter falls into the category of no CSIT because the transmitter needs no information of the exact channel coefficient values \cite{Jafar2012,Gou2011}.

When the channel coefficients of a channel block consisting of three time slots is given by the staggered pattern
\begin{center}
\begin{tabular}{|c|c|c|}
\hline
\cellcolor[gray]{0.9}$H_{1}'(\alpha)$ & \cellcolor[gray]{0.9}$H_{1}'(\alpha)$ & $H_{1}'(\beta)$\tabularnewline
\hline
$H_{2}'(\gamma)$ & \cellcolor[gray]{0.9}$H_{2}'(\phi)$ & \cellcolor[gray]{0.9}$H_{2}'(\phi)$\tabularnewline
\hline
\end{tabular}\quad,\\
\par\end{center}
we say the block channel has the type-$\overrightarrow{\mathsf{Z}}$
pattern.
Substituting the channel coefficients of the block into \eqref{eq:IA.signal.relation},
we have
\begin{equation}
\begin{cases}
\mathbf{u}_{1} & \rightarrow \mathsf{diag}[\frac{h_{11}'(\alpha)}{h_{21}'(\alpha)},\frac{h_{11}'(\alpha)}{h_{21}'(\alpha)},\frac{h_{11}'(\beta)}{h_{21}'(\beta)}]\mathbf{v}_{1}\\
\mathbf{u}_{2} & \rightarrow \mathsf{diag}[\frac{h_{12}'(\gamma)}{h_{22}'(\gamma)},\frac{h_{12}'(\phi)}{h_{22}'(\phi)},\frac{h_{12}'(\phi)}{h_{22}'(\phi)}]\mathbf{v}_{2}
\end{cases}.\label{eq:IA.blind}
\end{equation}
For this block pattern, if we choose $\mathbf{v}_{1}=[1,1,0]^{T}$ and $\mathbf{v}_{2}=[0,1,1]^{T}$,
then we get $\mathbf{u}_{1}\rightarrow[1,1,0]^{T}$ and $\mathbf{u}_{2}\rightarrow[0,1,1]^{T}$,
which achieves the interference alignment without the need of CSIT, i.e., BIA. Similarly, if a channel block has the pattern
\begin{center}
\begin{tabular}{|c|c|c|}
\hline
$H_{1}'(\alpha)$ & \cellcolor[gray]{0.9}$H_{1}'(\beta)$ & \cellcolor[gray]{0.9}$H_{1}'(\beta)$\tabularnewline
\hline
\cellcolor[gray]{0.9}$H_{2}'(\gamma)$ & \cellcolor[gray]{0.9}$H_{2}'(\gamma)$ & $H_{2}'(\phi)$\tabularnewline
\hline
\end{tabular}\quad,\\
\par\end{center}
we say the block has the type-$\overleftarrow{\mathsf{Z}}$ pattern.
For this block pattern, we can choose $\mathbf{v}_{1}=[0,1,1]^{T}$ and $\mathbf{v}_{2}=[1,1,0]^{T}$,
and then choose $\mathbf{u}_{1}\rightarrow[0,1,1]^{T}$ and $\mathbf{u}_{2}\rightarrow[1,1,0]^{T}$, which also obtains BIA.
Together, we refer to the two block channel patterns above as the type-$\mathsf{Z}$ pattern.

To show how the knowledge of the coherent time and the offset will help
implement BIA, we give a simple demonstration. Consider a BC channel with $N=5$ and $n_{\delta,2}=2$; its channel coefficients over time are plotted in Fig.~\ref{table:1}, in which the same marks represent the same channel state.
\begin{figure*}
\begin{center}
\begin{tabular}{|c||c|c|c|c|c|c|c|c|c|c|c|c|c|c|c|c|c|c|}
\hline
$n$ & 0 & 1 & 2 & 3 & 4 & 5 & 6 & 7 & 8 & 9 & 10 & 11 & 12 & 13 & 14 & 15 & 16 & 17\tabularnewline
\hline
$H_{1}(n)$ & $\Box$ & $\Box$ & $\Box$ & $\Box$ & $\Box$ & $\circ$ & $\circ$ & $\circ$ & $\circ$ & $\circ$ & $\bigtriangleup$ & $\bigtriangleup$ & $\bigtriangleup$ & $\bigtriangleup$ & $\bigtriangleup$ & $\heartsuit$ & $\heartsuit$ & $\heartsuit$\tabularnewline
\hline
$H_{2}(n)$ & $\bigstar$ & $\bigstar$ & $\clubsuit$ & $\clubsuit$ & $\clubsuit$ & $\clubsuit$ & $\clubsuit$ & $\blacklozenge$ & $\blacklozenge$ & $\blacklozenge$ & $\blacklozenge$ & $\blacklozenge$ & $\spadesuit$ & $\spadesuit$ & $\spadesuit$ & $\spadesuit$ & $\spadesuit$ & $\bullet$\tabularnewline
\hline
\end{tabular}
\end{center}
\caption{A homogeneous block fading channel with $N=5$ and $n_{\delta,2}=2$.}
\label{table:1}
\end{figure*}
The channel block from $n=3$ to $n=17$, which contains $3N=15$ consecutive time slots, can be decomposed
into 5 type-$\mathsf{Z}$ patterns as presented in Fig.~\ref{table:2}.
\begin{figure*}
\begin{center}
\begin{tabular}{|c|c|c|}
\hline
3 & 5 & 7\tabularnewline
\hline
$\Box$ & $\circ$ & $\circ$\tabularnewline
\hline
$\clubsuit$ & $\clubsuit$ & $\blacklozenge$\tabularnewline
\hline
\end{tabular} \quad %
\begin{tabular}{|c|c|c|}
\hline
4 & 6 & 8\tabularnewline
\hline
$\Box$ & $\circ$ & $\circ$\tabularnewline
\hline
$\clubsuit$ & $\clubsuit$ & $\blacklozenge$\tabularnewline
\hline
\end{tabular} \quad %
\begin{tabular}{|c|c|c|}
\hline
9 & 11 & 13\tabularnewline
\hline
$\circ$ & $\bigtriangleup$ & $\bigtriangleup$\tabularnewline
\hline
$\blacklozenge$ & $\blacklozenge$ & $\spadesuit$\tabularnewline
\hline
\end{tabular} \quad %
\begin{tabular}{|c|c|c|}
\hline
10 & 12 & 15\tabularnewline
\hline
$\bigtriangleup$ & $\bigtriangleup$ & $\heartsuit$\tabularnewline
\hline
$\blacklozenge$ & $\spadesuit$ & $\spadesuit$\tabularnewline
\hline
\end{tabular} \quad %
\begin{tabular}{|c|c|c|}
\hline
14 & 16 & 17\tabularnewline
\hline
$\bigtriangleup$ & $\heartsuit$ & $\heartsuit$\tabularnewline
\hline
$\spadesuit$ & $\spadesuit$ & $\bullet$\tabularnewline
\hline
\end{tabular}\\
\par\end{center}
\caption{The implementation of BIA for the homogeneous block fading channel shown in Fig.~\ref{table:1}.}
\label{table:2}
\end{figure*}
For instance, the time slots $n_1=3$, $n_2=5$ and $n_3=7$ form a channel block with the type-$\overleftarrow{\mathsf{Z}}$ pattern. Since the channel repeats the same pattern every 15 consecutive symbols afterwards, the same decomposition
can be achieved repetitively. Therefore, this channel can obtain the optimal DoF by using BIA. Note that this decomposition of the block from $n=3$ to $n=17$ is not unique, and the first three time slots are negligible or utilized for control messages.

\section{Blind interference alignment over homogeneous block fading channels}

Define $\tau \triangleq \min\{n_{\delta,2},N-n_{\delta,2}\}$. We first study the impact of $\tau$ on the feasibility of the BIA for a $2$-user $2\times 1$ MISO BC channel.
In such a channel with homogeneous block fading, a sufficient condition on $\tau$ to enable BIA is given by the following theorem.
\begin{thm}
For a 2-user $2\times 1$ MISO BC channel with the homogeneous coherent time $N$, if $\left\lceil N/3\right\rceil \leq\tau\leq\left\lfloor N/2\right\rfloor $, or equivalently $\left\lceil N/3\right\rceil \leq n_{\delta,2} \leq N - \left\lceil N/3\right\rceil $, 
then the channel can be decomposed into blocks of the type-$\mathsf{Z}$
pattern, and thus, the optimal DoF $\frac{4}{3}$ is achievable
by BIA. \end{thm}
\begin{proof}
\begin{figure*}
\begin{centering}
\includegraphics[scale=0.65]{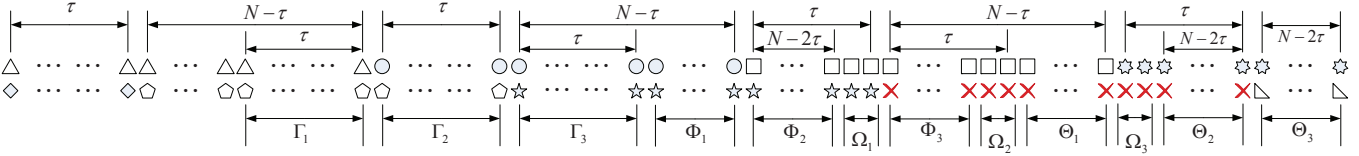}
\caption{Two downlink block fading links with the same coherent time $N$ and a time offset $\left\lceil N/3\right\rceil \leq\tau\leq\left\lfloor N/2\right\rfloor $ }
\label{figure:proof.thm01}
\end{centering}
\end{figure*}
We prove the theorem by showing a constructive algorithm which decomposes each $3N$ consecutive symbols into $N$ type-$\mathsf{Z}$ pattern blocks. A channel consisting of two block fading links with the time offset $\left\lceil N/3\right\rceil \leq\tau\leq\left\lfloor N/2\right\rfloor $ is plotted in Fig.~\ref{figure:proof.thm01}. The link on the top has $n_{\delta,1}=0$, and is denoted by first a coherent block of $N$ equilateral triangles representing unchanged channel states, followed by coherent blocks of $N$ symbols each. The link on the bottom has offset $\tau$, and is characterized by first $\tau$ rhombuses, followed by coherent blocks of $N$ symbols each.

Let 2-tuple $\{H_1(n),H_2(n)\}$ at time $n$ be referred to as a {\it pair}. Denote $\Gamma_1$ as the set of pairs which is formed by the last $\tau$ equilateral triangles on the top and the aligned pentagons on the bottom, $\Gamma_2$ as the set of pairs formed by the first $\tau$ circles on the top and their corresponding pentagons on the bottom, $\Gamma_3$ as the set of pairs formed by the first $\tau$ stars on the bottom and their corresponding circles on the top. Then, taking one pair from each of $\Gamma_i$, $i=1,2,3$, and grouping the selected three pairs will form a type-$\mathsf{Z}$ pattern block
\begin{center}
\begin{tabular}{|c|c|c|}
\hline
$\triangle$ & \Circle & \Circle \tabularnewline
\hline
\pentagon & \pentagon & $\bigstar$\tabularnewline
\hline
\end{tabular}.\\
\par\end{center}
Naturally, by using this construction method the pairs in $\Gamma_1$, $\Gamma_2$ and $\Gamma_3$ can form $\tau$ type-$\mathsf{Z}$ pattern blocks.

Similarly, three sets with $N-2\tau$ pairs each can be collected, i.e., $\Phi_1$, $\Phi_2$ and $\Phi_3$ as shown in Fig.~\ref{figure:proof.thm01}. The same constructive procedure over the three sets can generate $N-2\tau$ type-$\mathsf{Z}$ pattern blocks. Repeating the same approach, we group three sets with $3\tau - N$ pairs each as $\Omega_1$, $\Omega_2$ and $\Omega_3$; three sets with $N-2\tau$ pairs each as $\Theta_1$, $\Theta_2$ and $\Theta_3$. The constructive procedure will generate $3\tau-N$ and $N-2\tau$ type-$\mathsf{Z}$ pattern blocks from $\Omega_i$s and $\Theta_i$s, respectively.

The channel from $\Gamma_1$ to $\Theta_3$ forms a block of $3N$ symbols, and the channel pattern of the $3N$ symbols repeats over time. Every such a $3N$-symbol block can be decomposed into $N$ type-$\mathsf{Z}$ pattern blocks as described previously. Therefore, the channel with $\left\lceil N/3\right\rceil \leq\tau\leq\left\lfloor N/2\right\rfloor $ can be decomposed into type-$\mathsf{Z}$ pattern blocks, and thus achieves the optimal $\frac{4}{3}$ DoF by applying BIA to each type-$\mathsf{Z}$ pattern block.
\end{proof}

According to Theorem 1, the coherent block offset $\left\lceil N/3\right\rceil \leq\tau\leq\left\lfloor N/2\right\rfloor $ between two receivers is sufficient for BIA to achieve the optimal DoF. In practice, the channel coefficients within a coherent block are not exactly the same, but only highly correlated \cite{Rappaport2002}. There is no unique way to divide a block fading channel into coherent blocks over time. If the two users experience the same temporal correlated block fading, we can always divide the link of the second user into coherent blocks by selecting a $n_{\delta,2}$ satisfying $\left\lceil N/3\right\rceil \leq\tau\leq\left\lfloor N/2\right\rfloor $. Therefore, in practical systems a $2$-user $2\times 1$ MISO BC channel with homogeneous block fading can obtain the optimal $\frac{4}{3}$ DoF by BIA. Note that each true channel coefficient in a coherent block can be decoupled into the constant channel efficient characterizing the block plus a distortion. How the distortion affects the achieved DoF is interesting but beyond the scope of this paper.

Physically, it is justified to model $n_{\delta,2}$ as a uniformly distributed random variable over $0\leq n_{\delta,2}\leq N-1$. Then the chance for such a 2-user channel to achieve $\frac{4}{3}$ DoF by BIA is $1/3$. Now,
consider a $2\times 1$ BC network with $K\geq3$ mobile users experiencing homogeneous block fading. Suppose
the time offset of the link for $\mathbb{R}_{j}$ ($2\leq j\leq K$)
is uniformly distributed over $0\leq n_{\delta,j}\leq N-1$, with respect to $n_{\delta,1}=0$. Now, we examine \emph{how big the chance is
that the $2$-antenna transmitter can find two users among the $K$ users to form
a 2-user MISO BC channel which is able to achieve the optimal DoF $\frac{4}{3}$
by using BIA}. We first
present a lemma.
\begin{lem}\label{lm1}
Suppose there are $N$ boxes forming a ring. Cast a ball into one box and denote the box as the $0$th box, then cast other $K-1$ balls randomly into the $N$ boxes. Define the distance of any
two balls as the smallest number of boxes separating them along the ring. Let $f(N, K)$
be the number of the events in which the distance of any pair
of balls is less than $\lceil\frac{N}{3}\rceil$. Then
\begin{equation}\label{eq:f.N.K}
f(N,K)=3\Theta\left(\lceil \frac{N}{3} \rceil ,K-2\right)-2\Theta\left( \lceil \frac{N}{3} \rceil ,K-3\right)
\end{equation}
with $\Theta(a,b)=\sum_{i=1}^{a}i^{b}$. \end{lem}
\begin{proof}
Starting from the $0$th box, we clockwise number the boxes by integers in ascending order, until the $(\lceil N/3 \rceil -1)$th box is reached. In addition, still from the $0$th box, we counter-clockwise label the boxes by decreasing integers, until the label $-(\lceil N/3 \rceil -1)$ is reached.

We first count the events in which the distance between two ends of the shortest arc of boxes containing all the $K$ balls is $n$ ($1\leq n \leq \lceil N/3 \rceil -1$). Given the $0$th box being one end, the other end could be either the $n$th box or the $(-n)$th box, the next $K-2$ balls can be randomly assigned to the boxes between the $0$th box and the other end box, resulting in totally $2(n+1)^{K-2}$ events.  If the $i$th box ($1\leq i <  n$) is one end, then the other end is determined to be the $(i-n)$th box, the other $K-3$ balls (note that the $0$th box already has one ball) can be assigned in $(n+1)^{K-3}$ ways, and thus the number of events with such two ends are $(n-1)(n+1)^{K-3}$. Totally the number of the events with two ends distanced by $n$ is $2(n+1)^{K-2}+(n-1)(n+1)^{K-3} = (3n+1)(n+1)^{K-3}$.

By summing the number over $1\leq n \leq \lceil N/3 \rceil -1 $ and adding one event in which all balls are in the $0$th box, we obtain the total number of events in which any two balls are distanced by less than $\lceil N/3 \rceil$
\begin{align}
f(N,K)&=1 + \sum_{n=1}^{ \lceil N/3 \rceil -1} (3n+1)(n+1)^{K-3}  \nonumber \\
& =3\Theta\left(\lceil \frac{N}{3} \rceil ,K-2\right)-2\Theta\left( \lceil \frac{N}{3} \rceil ,K-3\right), \nonumber
\end{align}
which proves this lemma.
\end{proof}

From the lemma above, we can prove the following theorem.
\begin{thm}
Given the $2\times 1$ BC network with $K$ users, let $P(N,K)$ be the probability that the
transmitter finds two users among the $K$ users to form a 2-user $2\times 1$
MISO BC channel which can achieve the optimal DoF $\frac{4}{3}$ by using
BIA. Then
\begin{equation}\label{eq:P.N.K}
P(N,K)\geq1-\frac{3\Theta(\lceil \frac{N}{3}\rceil ,K-2)}{N^{K-1}}.
\end{equation}
\end{thm}
\begin{proof}
Let
\begin{equation}
\tau_{ij} = \min \{|n_{\delta,i} - n_{\delta,j}|, N-|n_{\delta,i} - n_{\delta,j}|\} \nonumber
\end{equation}
be the relative time offset between the user $i$ and the user $j$. Distributing $n_{\delta,j}$  ($2\leq j\leq K$) over $0\leq n_{\delta,j}\leq N-1$ uniformly is equivalent to casting $K-1$ balls uniformly into the ring of $N$ boxes. Furthermore, the event of $\tau_{ij}\leq \lceil N/3\rceil -1$ for all $i\neq j$ is equivalent to the event in which any two balls among $K$ balls in Lemma \ref{lm1} are distanced less than $\lceil \frac{N}{3} \rceil$. Therefore, the probability that $\tau_{ij}< \lceil N/3\rceil $ for all $i\neq j$ is $f(N,K)/N^{K-1}$. So, the probability that there exists $\tau_{ij}\geq \lceil N/3\rceil $ for some $i\neq j$, that is, the transmitter and the users $i$ and $j$ form a MISO BC channel achieving $\frac{4}{3}$ DoF by BIA, is
$P(N,K)=1-f(N,K)/N^{K-1}$. Finally, substituting \eqref{eq:f.N.K} into $P(N,K)$ proves the inequality \eqref{eq:P.N.K}.
\end{proof}

\section{Discussion and Conclusion}
\begin{figure}
\begin{centering}
\includegraphics[scale=0.7]{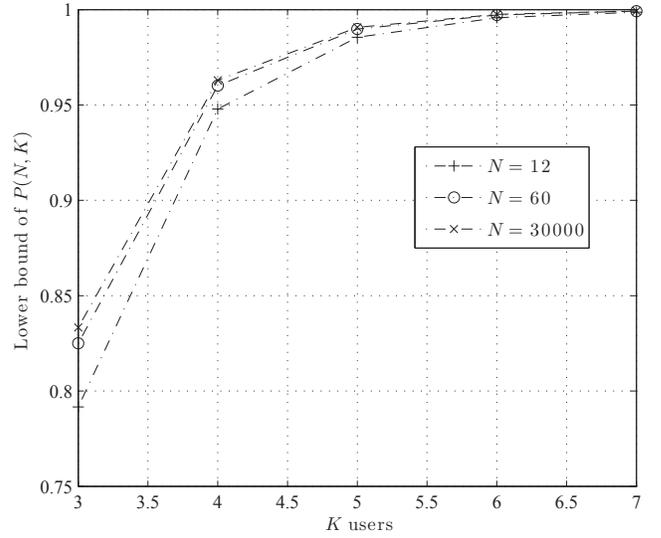}
\par\end{centering}
\caption{Lower bound of $P(N,K)$ as a function of user number $K$.}
\label{fig:lb.PNK}
\end{figure}
The lower bound of $P(N,K)$ in \eqref{eq:P.N.K} is plotted in Fig.~\ref{fig:lb.PNK} as a function of $K$ for $N=12$, $30$, $30000$. It is observed
the probability of finding 2 users among $K=4$ users to facilitate BIA
is $P(N,4)\geq 0.95$ for all $N\geq 12$. The chance of finding such two users among $K\geq 6$ users is almost one. It is also notable that the lower bound is not sensitive to $N$ when $N\geq 12$. Therefore, we conclude that regardless of the coherent time experienced by $K$ homogeneous users, if $K\geq 4$, it is almost guaranteed that the two-antenna transmitter can find two users among the $K$ users to form a BIA-favored 2-user $2\times 1$ MISO BC channel.

In summary, this paper shows that in a practical $2$-user $2\times 1$ MISO BC channel with homogeneous block fading, interference alignment can achieve the optimal $\frac{4}{3}$ DoF without the need of CSIT. Under the assumption that the coherent blocks of $K\geq 2$ homogeneous single-antenna users are randomly offset with uniform distribution, it is almost guaranteed that when $K\geq 4$, the transmitter can find two users to form a $2$-user $2\times 1$ BC channel, and achieve the optimal $\frac{4}{3}$ DoF by BIA. The results in this paper are also applicable to the $2\times 2$ X channel since the $2$-user $2\times 1$ MISO BC channel has the same system model as the $2\times 2$ X channel, in which two single-antenna transmitters broadcast to two single-antenna receivers.

\end{document}